\documentclass[preprint, final, 12pt,times]{elsarticle}

\usepackage{amssymb}

 \usepackage{amsthm}

\usepackage[all]{xy}
\usepackage{mathtools}
\newtheorem{lemma}{Lemma}[section]
\newtheorem{proposition}{Proposition}[section]
\newtheorem{theorem}{Theorem}[section]

\newtheorem{remark}{Remark}[section]

\usepackage[utf8]{inputenc}
\usepackage[T1]{fontenc}

\usepackage{cancel}
\usepackage{ulem}

\usepackage{theoremref}

\usepackage{tikz-cd}
\usepackage{enumitem}
\usepackage{calrsfs}
\DeclareMathAlphabet{\pazocal}{OMS}{zplm}{m}{n}
\usepackage{pgfplots}

 \usepackage{lineno}

\makeatletter
\def\ps@pprintTitle{%
   \let\@oddhead\@empty
   \let\@evenhead\@empty
   \let\@oddfoot\@empty
   \let\@evenfoot\@oddfoot
}
\makeatother

\begin{document}

\begin{frontmatter}

\title{Well-Rounded Lattices via Polynomials \tnoteref{label}}

\author{Carina Alves }
\tnotetext[label]{This work has been supported by FAPESP grants 2013/25977-7 and 2018/12702-3\\
carina.alves@unesp.br, williamlima.unesp.rc@gmail.com, antonio.andrade@unesp.br}

\author{William Lima da Silva Pinto}

\author{Antonio Aparecido de Andrade}
%

\begin{abstract}
Well-rounded lattices  have been a topic of recent studies with applications in wiretap channels and in cryptography. 
A lattice of full rank in Euclidean space is called well-rounded
if its set of minimal vectors spans the whole space. In this paper, we investigate when  lattices coming from polynomials with integer coefficients are well-rounded. 
\end{abstract}


\end{frontmatter}


\section{Introduction}

A large class of the problems in coding theory is related to the properties of lattices \cite{toni, camp, joao}.  A \textit{lattice} $\Lambda$ is a discrete additive subgroup of $\mathbb{R}^n.$ Equivalently, $\Lambda \subset \mathbb{R}^n$ is a lattice if  there are linearly independent vectors $v_1,\cdots, v_m\in \mathbb{R}^n$, with $m\leq n,$ such that any $x\in \Lambda$ can be written as $x=\displaystyle\sum_{i=1}^{m}x_iv_i$, where $ x_i\in \mathbb{Z}.$ The set $\{v_1,\cdots, v_m\}$ is called a basis for $\Lambda. $ A matrix $M$ whose rows are these vectors is said to be a \textit{generator matrix} for $\Lambda$ and its \textit{Gram matrix} is $G=MM^t,  $ where $^t$ stands for the transpose. If $m=n$, then $\Lambda$ is a \textit{full-ranked lattice}. The \textit{determinant} of $\Lambda$ is given by $det(\Lambda)=det(G)$ and it is an invariant under basis change \cite{conway}.  

The utility of a given lattice for a given application is measured using some relevant invariants, such as the packing density, minimal vectors, etc.
The \textit{minimum} of a lattice $\Lambda$ is defined by
$| \Lambda | = min\{{\Vert x \Vert}^2:x\in \Lambda, \ x\neq 0\}$ and its \textit{center density} is  $\delta(\Lambda)=\frac{(\sqrt{|\Lambda|}/2)^n}{|det(M)|}$.
The set of minimal vectors of $\Lambda$ is defined by $
S(\Lambda)= \{x \in \Lambda:\Vert x \Vert ^2=| \Lambda |\}
$ and its elements are called \textit{minimal vectors} of $\Lambda$. 
We say a lattice $\Lambda$ is well-rounded when $S(\Lambda)$ spans $\mathbb{R}^n$.

Well-rounded lattices and configurations of their minimal vectors play an important role in discrete optimization problems (see \cite{conway, martinet}). In particular, spherical configurations which give good kissing
numbers always come from well-rounded lattices.
In \cite{fukshansky1, fukshansky3}, the authors investigated the connection between well-rounded lattices and the well known ideal lattice, focusing especially on the case of lattices in $\mathbb{R}^2$. 
Another interesting class of lattices in dimension $2$ and $3$ was introduced in \cite{andre}. It essentially consists of constructing a generator matrix from the set of roots of a polynomial.  
Inspired by it and due to the importance of well-rounded lattices, in this paper we investigate in which conditions lattices obtained by polynomials up to dimension $4$ are well-rounded.  In this process, the use of the well known Vieta's formulas  will be essential since they simplify substantially the expression which gives the square distance of $x\in \Lambda$ to the origin.

An advantage of studying well-rounded lattices obtained via polynomials is that, in this process, we are able to establish conditions over a polynomial in such way that  we increase the number of minimum vectors, giving rise to lattices with higher packing densities. In this paper, we find lattices with  the highest packing densities in dimensions $2$ and $3$.
It is shown in  \cite{fukshansky1} that only ideal lattices coming from cyclotomic fields are well-rounded. A particular case of the constructions presented here is that well-rounded lattices $\Lambda_f$ can also be obtained from  totally real number fields $\mathbb{K}$ generated by a root of an irreducible polynomial $f(x)$ over $\mathbb{Q}$ whose roots are real and some restrictions are made to the coefficients of $f(x)$. The generator matrix of $\Lambda_f$ consists of the roots of $f(x).$ More recently, lattices arising from totally real number fields have found applications in communication over wireless channels \cite{oggier}.
 However, more general constructions are presented here and $f(x)$ is not necessarily irreducible.

This paper is organized as follows. In Section II, we present conditions to obtain well-rounded lattices via quadratic polynomials with distinct real roots and complex conjugate roots. We also obtain lattices with the highest center density in dimension $2$. In Section III,  we present conditions to obtain well-rounded lattices via cubic polynomials. In dimension $3$, we obtain lattices with the highest center density. In Section IV, we present conditions to obtain well-rounded lattices via quartic polynomials.
 
\section{Well-rounded lattices via quadratic polynomials} 

In \cite{andre} it is shown that monic polynomials of degree $2$ with distinct real and complex roots yield lattices with the highest center density in dimension $2$. Thus, we consider here both cases. 

\subsection{Distinct real roots}

Let $f(x)=x^2+ax+b\in\mathbb{Z}[x]$ be a polynomial with two distinct real roots denoted by $\alpha$ and  $\beta$. In this case, the discriminant  is greater than zero, i.e, $a^2-4b>0$. Moreover, let us construct a lattice $\Lambda_f$ by identifying a linearly independent set $\{v_1, v_2\}$ over $\mathbb{R},$ where $v_i$ is in terms of the roots of $f(x),$ $i=1,2.$
 In order to choose linearly independent vectors $v_1$ and $v_2$ is enough to ensure that the matrix $M$ from these vectors has nonzero determinant. The matrix $M$ will be then a generator matrix of $\Lambda_f$.
 
Consider $v_1=(\alpha,\beta)$ and $v_2= (\beta,\alpha)$. Thus,
 $$det(M)=\begin{vmatrix}
 \alpha & \beta \\
 \beta & \alpha
 \end{vmatrix}
 = -2a\sqrt{a^2-4b}\neq 0,
 $$
 since $a\neq 0.$  Therefore, if we consider $f(x)=x^2+ax+b\in\mathbb{Z}[x]$, ith $a\neq 0$, then we can define $\Lambda_f$ as the lattice generated by $v_1=(\alpha,\beta)$ and $v_2= (\beta,\alpha)$. 
 
 The next result give us an expression for the minimum of $\Lambda_f.$

 \begin{lemma}\label{normavetor2}
 Let $f(x)=x^2+ax+b \in \mathbb{Z}[x]$, with $a\neq 0$,  be a polynomial with real distinct roots $\alpha$, $\beta$ and $\Lambda_f$ be the lattice generated by $\{v_1,v_2\}$, where $v_1=(\alpha,\beta)$ and $v_2=(\beta,\alpha)$. If   $x=x_1v_1+x_2v_2$ is a point of $\Lambda_f$, where 
  $x_1,x_2\in\mathbb{Z}$, then  $\Vert x \Vert^2=a^2({x_1}^2+{x_2}^2)-2b(x_1-x_2)^2.$
\end{lemma} 
\begin{proof} We have  that $\Vert x \Vert^2=x\cdot x=x_1^2v_1^2+2x_1x_2v_1v_2+x_2^2v_2^2.$ 
Since $v_i^2=\alpha^2+\beta^2=a^2-2b$, for $i=1,2,$ and $v_1\cdot v_2=2\alpha\beta=a^2-2b$ the result follows. 
\end{proof}

  We are interested in verifying in which conditions $\Lambda_f$ is well-rounded.
For notation purposes, define $d:\mathbb{Z}^2 \rightarrow \mathbb{R}_{+}$ by $d(x_1, x_2)=a^2(x_1^2+x_2^2)-2b(x_1-x_2)^2$, i.e., $\Vert x \Vert^2=d(x_1,x_2)$. 
To identify when $d(x_1,x_2)$ assumes minimum value  it is easy to see that is enough to check its values when $x_1$ and $x_2$ vary between $0$, $1$ and $-1$. Note that
\begin{itemize}
\item[(i)] $d(\pm 1,0)=d(0,\pm 1)=a^2-2b$;
\item[(ii)] $d(\pm 1,\pm 1)=2a^2$;
\item[(iii)] $d(-1,1)=d(1,-1)=2a^2-8b$.
\end{itemize}

We know from \cite{fukshansky1} that $\Lambda_f \subset \mathbb{R}^2$ is well-rounded if and only if $|S(\Lambda)|=4$ or $|S(\Lambda)|=6$.  In this later case,  $\Lambda_f$ has the highest packing density in dimension $2$.
Note that $|S(\Lambda)|=4$ if and only if $min\{a^2-2b, 2a^2,2a^2-8b\}=a^2-2b,$ that is, $a^2-2b\leq a^2$ and $a^2-2b\leq 2a^2-8b.$ This means that $-2b\leq a^2$ and $6b\leq a^2.$
When  $b\geq 0$ ($b<0$) the inequalities above are satisfied if and only if $a^2\geq 6b$ ($a^2\geq -2b$).
It is easy to see that $|S(\Lambda)|=6$ if and only if $a^2-2b=2a^2$ or $2a^2=2a^2-8b$ which is equivalent to $a^2=-2b$ or $a^2=6b.$
 Consequently, we have proved that the following theorem holds true.
 
\begin{theorem}\thlabel{principalgrau2}
Let $f(x)=x^2+ax+b$, where $a,b\in\mathbb{Z}$, with $a\neq 0$, be a polynomial with real distinct roots. 
 If $\alpha,\beta\in\mathbb{R}$ are the distinct roots of $f(x)$, then the lattice 
$\Lambda_f$ generated by the basis $\{ (\alpha,\beta),(\beta,\alpha )\}$ is well-rounded if and only if $a^2\geq -2b$ ($b<0$) or $a^2 \geq 6b$ ($b\geq 0$). Moreover, $\Lambda_f$ has  the highest packing density in dimension 2 if and only if $a^2=6b$ or $a^2=-2b$.
\end{theorem}

\subsection{Complex conjugate roots}

We consider $f(x)$ as above, but with two complex conjugate roots denoted by $\gamma_1=\alpha +i \beta$ and $\gamma_2=\alpha-i\beta$. In this case, its  discriminant is less than zero, i.e., $a^2-4b<0.$  
Note that we can define a lattice $\Lambda_f$ as the one generated by $v_1=(\alpha,\beta)$ and $v_2= (\alpha, -\beta)$ since 
 $$det(M)=\begin{vmatrix}
 \alpha & \beta \\
 \alpha & -\beta
 \end{vmatrix}
 = -2\alpha\beta=\pm\frac{1}{2}a\sqrt{-a^2+4b}\neq 0,
 $$
where in the last equality we use the Vieta's formula, that is, that  $\gamma_1+\gamma_2=-a$ and $\gamma_1\gamma_2=b.$

The next result give us an expression for the norm of a vector in  $\Lambda_f$.

\begin{lemma}\label{normavetcomplexa} Let $f(x)=x^2+ax+b \in \mathbb{Z}[x],$ with $a\neq 0$, be a polynomial with complex conjugate roots $\alpha\pm i\beta,\ \alpha,\beta\in\mathbb{R}, \ \beta\neq 0$ and  $\Lambda_f$ be a lattice  generated by  $\{v_1,v_2\}$, where $v_1=(\alpha,\beta)$ and  $v_2=(\alpha,-\beta)$. If $x=x_1v_1+x_2v_2$ is a point of $\Lambda_f$, where  $x_1,x_2\in\mathbb{Z}$, then $\Vert x \Vert^2=\frac{a^2}{4}(x_1+x_2)^2+\frac{4b-a^2}{4}(x_1-x_2)^2.$
\end{lemma} 
\begin{proof}  The result follows by applying the Vieta's formula in $f(x).$
\end{proof}

\begin{theorem}
Let  $f(x)=x^2+ax+b\in\mathbb{Z}[x]$, with $a\neq 0$, be a polynomial with complex conjugate roots . If $\alpha\pm i\beta$ are the roots of $f(x)$, then the lattice $\Lambda_f$ generated by the  basis $\{ (\alpha,\beta),(\alpha,-\beta ) \}$ is well-rounded if and only if $b\leq a^2\leq 3b$. Moreover, $\Lambda_f$ has the highest packing density   in dimension $2$ if and only if $a^2=b$ or $a^2=3b$.
\end{theorem}
\begin{proof}
Define $d:\mathbb{Z} \rightarrow \mathbb{R}_{+}$ by $d(x_1,x_2)=\frac{a^2}{4}(x_1+x_2)^2+\frac{4b-a^2}{4}(x_1-x_2)^2$, i.e., $\Vert x \Vert^2=d(x_1,x_2)$. Similarly, to what we have done in Theorem \ref{principalgrau2}, we will check the value of $d(x_1,x_2)$ varying $x_1$ and $x_2$ between $-1$, $0$ and $1$. We have that
\begin{itemize}
\item[(i)] $d(\pm 1,0)=d(0,\pm 1)=b$;
\item[(ii)]$d(\pm 1, \pm 1)=a^2$;
\item[(iii)]$d(1,-1)=d(-1,1)=4b-a^2$.
\end{itemize}
Again, by \cite{fukshansky1}, it follows that  $\Lambda_f$ is well-rounded if and only if $|S(\Lambda_f)|=4$ or $|S(\Lambda_f)|=6$.  
Note that $|S(\Lambda)|=4$ if and only if $min\{b, a^2,4b-a^2\}=b,$ that is, $b\leq a^2$ and $a^2\leq 3b.$ 
When  $b> 0$ we conclude that $|S(\Lambda)|=4$ if and only if 
$b\leq a^2 \leq 3b.$ 
It is clear  that $|S(\Lambda_f)|=6$ if and only if $b=a^2$ or $a^2=3b.$ In this case,  $\Lambda_f$ has the highest packing density in dimension $2$.
\end{proof}

\section{Well-rounded lattices via cubic polynomials}

Let $f(x)=x^3+ax^2+bx+c$ be a polynomial with integers coefficients with  three  distinct real roots denoted by $\alpha,\,\beta$ and  $\gamma$. 
 Let us construct a lattice $\Lambda_f$ by identifying a linearly independent set $\{v_1, v_2, v_3\}$ over $\mathbb{R},$ where $v_i$ are given in terms of the roots of $f(x),$ for $i=1,2,3$.
  In order to choose linearly independent vectors $v_1, v_2$ and $v_3$ is enough to ensure that the matrix $M$ whose rows are these vectors has nonzero determinant. The matrix $M$ will be then a generator matrix to $\Lambda_f.$
 Consider $v_1=(\alpha,\beta,\gamma)$,  $v_2= (\gamma,\alpha,\beta)$ and $v_3=(\beta,\gamma,\alpha)$. A simple calculation shows that
$$det(M)=\begin{vmatrix}
        \alpha & \beta  & \gamma \\
        \gamma & \alpha & \beta  \\
        \beta  & \gamma & \alpha \\
 \end{vmatrix}
 = -a(a^2-3b).
 $$
Since $f(x)$ has real distinct roots, it follows that $f'(x)$ also has, which happens if its discriminant is greater than zero, that is,  $a^2-3b>0$.
Therefore, if $a\neq 0$, then $det(M) \neq 0$.

The next result give us an expression to the minimum of $\Lambda_f$.

\begin{lemma} \cite{andre} \label{lemanorm3}
 Let $f(x)=x^3+ax^2+bx+c \in\mathbb{Z}[x]$, with $a\neq 0$, be a polynomial with  distinct real roots
 $\alpha, \beta$ and $\gamma.$ Let $\Lambda_f$ be a lattice in $\mathbb{R}^3$ generated by the basis $\{v_1,v_2,v_3\}$, where  $v_1=(\alpha,\beta,\gamma)$, $v_2=(\gamma,\alpha,\beta)$ and $v_3=(\beta,\gamma,\alpha)$. If $x \in \Lambda_f$, where $x=x_1v_1+x_2v_2+x_3v_3$, with $x_1,x_2,x_3\in\mathbb{Z}$, then
  \[\Vert x\Vert ^2=(a^2-2b)(x_1^2+x_2^2+x_3^2) + 2b(x_1x_2+x_1x_3+x_2x_3).\]
\end{lemma}

Now, since we have an expression to the minimum norm of $\Lambda_f$ the next theorem tells us when $\Lambda_f$ is well-rounded.

\begin{theorem}\label{teograu3}
Let  $f(x)=x^3+ax^2+bx+c\in\mathbb{Z}[x]$, with $a\neq 0.$  If $\alpha,\beta,\gamma\in\mathbb{R}$ are the distinct real roots of $f(x)$, then the lattice $\Lambda_f$ generated by the basis $\{(\alpha,\beta,\gamma),(\gamma,\alpha,\beta),(\beta,\gamma,\alpha)\}$ is well-rounded if and only if $a^2\geq 4b$ ($b\geq 0$) or  $a^2\geq -b$ ($b<0$). Moreover, $\Lambda_f$ has the highest packing density for dimension $3$ if and only if $a^2=4b$.
\end{theorem}
\begin{proof}
Define  $d:\mathbb{Z}^2 \rightarrow \mathbb{R}_{+}$ by $d(x_1,x_2,x_3)=(a^2-2b)(x_1^2+x_2^2+x_3^2)+2b(x_1x_2+x_1x_3+x_2x_3)$, i.e., $\Vert x \Vert^2=d(x_1,x_2, x_3)$.
To identify when $d(x_1,x_2,x_3)$ assumes minimum value it is easy to see that is enough to check its value when $x_1,\,x_2$ and $x_3$ vary between $0$, $1$ and $-1$. We have that 
\begin{itemize}
\item[(i)] $d(\pm 1,0,0)=d(0,\pm 1,0)=d(0,0,\pm 1)=a^2-2b$; 
\item[(ii)] $d(\pm 1,\pm 1,0)=d(\pm 1,0,\pm 1)=d(0,\pm 1,\pm 1)=2a^2-2b$; 
\item[(iii)] $d(1,-1,0)=d(-1,1,0)=d(1,0,-1)=d(-1,0,1)=d(0,1,-1)=d(0,-1,1)=2a^2-6b$; 
\item[(iv)] $d(\pm 1,\pm 1,\pm 1)=3a^2$;
\item[(v)] $d(1,1,-1)=d(1,-1,1)=d(-1,1,1)=d(-1,-1,1)=d(-1,1,-1)=d(1,-1,-1)=3a^2-8b$.
\end{itemize}
Note that if $a^2-2b=min\{a^2-2b,2a^2-2b,2a^2-6b,3a^2,3a^2-8b\}$, then $\Lambda_f$ is well-rounded, since the vectors  $x\in \Lambda_f$ such that $\Vert x \Vert ^2=a^2-2b$ are linearly independent. 
Let us show that $a^2-2b=min\{a^2-2b,2a^2-2b,2a^2-6b,3a^2,3a^2-8b\}$ is also a necessary condition for $\Lambda_f$ to be well-rounded. 
Suppose that $\Lambda_f$ is well-rounded and  $a^2-2b\neq min\{a^2-2b,2a^2-2b,2a^2-6b,3a^2,3a^2-8b\}$. Clearly $b\neq 0.$ Moreover, $3a^2-8b$ since $3a^2-8b=(a^2-2b)+(2a^2-6b)>min\{a^2-2b,2a^2-2b,2a^2-6b,3a^2,3a^2-8b\}$.
When $b>0$, it follows that $2a^2-6b\leq 2a^2-2b,$ $2a^2-6b\leq 3a^2-8b,$ $2a^2-6b<3a^2.$ By our assumption $2a^2-6b<a^2-2b$
and then we conclude that $2a^2-6b=min\{a^2-2b,2a^2-2b,2a^2-6b,3a^2,3a^2-8b\}$. Since  the vectors  $x\in \Lambda_f$ such that $\Vert x \Vert ^2=2a^2-6b$ are linearly dependent, it follows that $S(\Lambda_f)$ does not span $\mathbb{R}^3.$
When $b<0$, it follows that  $2a^2-2b<2a^2-6b<3a^2-8b$ and $a^2-2b<2a^2-2b.$ By our assumption, we conclude that
 $min \{a^2-2b,2a^2-2b,2a^2-6b,3a^2,3a^2-8b\}=3a^2$. Again, it is easy to see that in this case $S(\Lambda_f)$ does not span $\mathbb{R}^3.$
 Thus, $\Lambda_f$ is well-rounded if and only if $a^2-2b= min\{a^2-2b,2a^2-2b,2a^2-6b,3a^2,3a^2-8b\}$.  It means that  
$a^2-2b\leq 2a^2-2b$,  $a^2-2b\leq 2a^2-6b$, $a^2-2b\leq 3a^2$, $a^2-2b\leq 3a^2-8b$. 
When $b\geq 0$ ($b<0$)  the inequalities above are satisfied if and only if $a^2\geq 4b$ ($a^2\geq -b$).
When $a^2=-b$ or $a^2=4b$ we see that $|S(\Lambda_f)|$ increases. It means that $\Lambda_f$ has higher center density when these equalities are satisfied. Note that if $a^2=-b\, (b<0)$, then $\delta(\Lambda_f)=\frac{(\sqrt{3(-b)}/2)^3}{a(4b)}=\frac{3\sqrt{3}}{32}\approx  0.16238.$ On the other hand, if $a^2=4b$ ,then $\delta(\Lambda_f)=\frac{(\sqrt{4b}/2)^3}{|ab|}=\frac{1}{4\sqrt{2}}\approx 0.17679$ corresponding to the highest center density in  dimension $3$.
\end{proof}

\section{Well-rounded lattices via quartic polynomials}

We  consider here a monic polynomial of degree $4$ with integer coefficients and real roots. We will restrict our investigation to polynomials with two roots opposite roots of each other. The reason behind of it  will be detailed in the proposition that come next.

\begin{proposition}\label{formuladetG3}
Let  $f(x)=x^4+ax^3+bx^2+cx+d\in \mathbb{Z}[x],$ with $a\neq 0$, be a polynomial  with real and distinct roots $\alpha,\beta,\gamma,\psi$, such that $\alpha=-\gamma$. Under these conditions, the vectors $(\alpha,\beta,\gamma,\psi)$, $(\psi,\alpha,\beta,\gamma)$ $(\gamma,\psi,\alpha,\beta)$, and
 $(\beta,\gamma,\psi,
\alpha)$ are linearly independent.
\end{proposition}
\begin{proof} Since $\alpha=-\gamma$, it follows that $2\gamma=\gamma-\alpha$.  Moreover,  the  Vieta's formula implies that  $\beta+\psi=-a,$  $-\gamma^2+\beta\psi=b,$ $-\beta\gamma^2-\gamma^2\psi=-c$ and $-\beta\gamma^2\psi=d.$
Consider  the matrix $M$ whose rows are the vectors above, that is,
\[
M=
\begin{pmatrix}
-\gamma & \beta & \gamma & \psi \\
\psi & -\gamma& \beta & \gamma \\
\gamma & \psi & -\gamma & \beta \\
\beta &\gamma & \psi & -\gamma \\
\end{pmatrix}.
\]
In what follows, we are going to prove that $det(M) \neq 0$ making use of  the formulas above.
In fact,
\begin{align}\label{detG3}
det(M)& =(-\beta^4-\psi^4)+2\beta^2\psi^2-4\beta^2\gamma^2-4\gamma^2\psi^2-8\beta\gamma^2\psi \nonumber \\
&= -(\beta^2-\psi^2)^2-4\gamma^2(\beta^2+\psi^2+2\beta\psi)\nonumber \\
& = -(\beta-\psi)^2(\beta+\psi)^2-4\gamma^2(\beta+\psi)^2\nonumber \\
& = -(\beta+\psi)^2(4\gamma^2+(\beta-\psi)^2)\nonumber \\
&=-a^2(4\gamma^2+(\beta-\psi)^2)\nonumber \nonumber \\
& =-a^2((\gamma-\alpha)^2+(\beta-\psi)^2) \nonumber \\
&=-a^2(\gamma^2-2\gamma\alpha+\alpha^2+\beta^2-2\beta\psi +\psi^2)\nonumber \nonumber \\
& =-a^2(\alpha^2+\beta^2+\gamma^2+\psi^2-2(\gamma\alpha+\beta\psi))\nonumber \nonumber \\
&=-a^2(2\gamma^2+(\beta+\psi)^2-2\beta\psi-2(-\gamma^2+\beta\psi))\nonumber \\
&=-a^2(2\gamma^2+a^2-2\beta\psi+2\gamma^2-2\beta\psi)\nonumber \\
&=-a^2(4\gamma^2+a^2-4\beta\psi)\nonumber \nonumber \\
&=-a^2(4(\gamma^2-\beta\psi)+a^2)\nonumber \nonumber \\
&=-a^2(-4b+a^2).
\end{align}
Now, let us prove that $-a^2(-4b+a^2)\neq 0$. 
Suppose, for contradiction, that $-a^2(-4b+a^2)=0$. Since $a\neq 0$, it follows that $-4b+a^2=0.$ Replacing the Vieta's formula, the last equality implies that $4(\gamma^2-\beta\psi)+(\beta+\psi)^2=0,$ i.e., $(\beta-\psi)^2=-4\gamma^2,$ which is a contradiction.  
  Therefore, $det(M)=-a^2(-4b+a^2)\neq 0$ and the result follows. 
\end{proof} 

\begin{remark} It is important to note that the hypothesis $\alpha=-\gamma$ is essential to determining $det(M)$ in terms of the coefficients of $f(x).$
 The same remark can be done about the calculation of norm of a vector in $\Lambda_f$, as we see next. 
 Therefore, the hypothesis $\alpha=-\gamma$ is crucial and unique to our  approach.
 \end{remark}

According to Proposition \ref{formuladetG3},  we can consider a lattice in $\mathbb{R}^4$ generated by the  linearly independent vectors $v_1=(\alpha,\beta,\gamma,\psi),$ $v_2=(\beta,\gamma,\psi,\alpha),$ $v_3=(\gamma,\psi,\alpha,\beta)$ and $v_4=(\psi,\alpha,\beta,\gamma)$ over $\mathbb{R},$ where $\alpha,\beta,\gamma,\psi$ are the roots of $x^4+ax^3+bx^2+cx+d\in \mathbb{Z}[x],$ with $a\neq 0$, such that  $\alpha=-\gamma$.
 
\begin{lemma}\thlabel{norma4}
Let  $f(x)=x^4+ax^3+bx^2+cx+d\in \mathbb{Z}[x],$ with $a\neq 0$, be a polynomial  with distinct real roots $\alpha,\beta,\gamma,\psi$, such that  $\alpha=-\gamma$ and $\Lambda_f$ be a lattice generated by $\{v_1,v_2,v_3,v_4\},$ where $v_1=(\alpha,\beta,\gamma,\psi),$ $v_2=(\beta,\gamma,\psi,\alpha),$ $v_3=(\gamma,\psi,\alpha,\beta)$ and $v_4=(\psi,\alpha,\beta,\gamma).$ If $x=x_1v_1+x_2v_2+x_3v_3+x_4v_4$ is a point of $\Lambda_f,$ where $x_1, x_2,x_3, x_4\in \mathbb{Z}$, then 
 $\Vert x \Vert^2=(a^2-2b)(x_1^2+x_2^2+x_3^2+x_4^2)+4b(x_1x_3+x_2x_4)$.
\end{lemma}
\begin{proof} It is  easy to see that  $\Vert x \Vert^2=(x_1\alpha+x_2\beta+z_3\gamma+z_4\psi)^2+(x_1\beta+x_2\gamma+z_3\psi+z_4\alpha)^2+(x_1\gamma+x_2\psi+z_3\alpha+z_4\beta)^2+(x_1\psi+x_2\alpha+z_3\beta+z_4\gamma)^2$. Rearranging, it follows that
\begin{equation}\label{1norma4}
\begin{split}
\Vert x \Vert^2&=(\alpha^2+\beta^2+\gamma^2+\psi^2)(x_1^2+x_2^2+x_3^2+x_4^2)+\\
&+2(\alpha\beta+\alpha\psi+\beta\gamma+\gamma\psi)(x_1x_2+x_1x_4+x_3x_4+x_2x_3)+4(\alpha\gamma+\beta\psi)(x_1x_3+x_2x_4).
\end{split}
\end{equation}
Since  $\alpha=-\gamma$, by Vieta's formulas, it follows that $\alpha\beta+\alpha\psi+\beta\gamma+\gamma\psi=0$ and $\alpha\gamma+\beta\psi=b$. Moreover, $\alpha^2+\beta^2+\gamma^2+\psi^2=2\gamma^2+\beta^2+\psi^2=2\gamma^2-2\beta\psi+(\beta+\psi)^2=-2b+a^2.$
 Consequently, the Equation (\ref{1norma4}) implies $\Vert x \Vert^2=(a^2-2b)(x_1^2+x_2^2+x_3^2+x_4^2)+4b(x_1x_3+x_2x_4)$.
\end{proof}

We are interested in verifying in which conditions $\Lambda_f$ is well-rounded. To do it,  
 define $d:\mathbb{Z}^4\rightarrow \mathbb{R}_{+}$ by $d(x_1,x_2,x_3, x_4)=(a^2-2b)(x_1^2+
 x_2^2+x_3^2+x_4^2)+4b(x_1x_3+x_2x_4)$, i.e., $d(x_1,x_2,x_3,x_4)=\Vert x \Vert^2$. To identify when $d(x_1,x_2,x_3, x_4)$ assumes minimum value it is easy to see that is enough to check its values when $x_i$ varies between $0$, $1$ and $-1$, for $i\in\{1,2,3,4 \}$. In particular,

\begin{itemize}
\item[(i)] $d(\pm1,0,0,0)=d(0,\pm1,0,0)=d(0,0,\pm1,0)=d(0,0,0,\pm1)=a^2-2b$;
\item[(ii)] $d(\pm1,0,\pm 1,0)=d(0,\pm1,0,\pm1)=2a^2$;
\item[(iii)] $d(1,0,-1,0)=d(-1,0,1,0)=d(0,1,0,-1)=d(0,-1,0,1)=2a^2-8b$.
\end{itemize}
Set  $A=a^2-2b$, $B=2a^2$ and $C=2a^2-8b$. The remaining possibilities for $d(x_1,x_2,x_3,x_4)$, with $x_i$ varying between $0$,$1$ and $-1$ for all $i\in\{1,2,3,4\}$, are all greater than $A$, $B$ or $C$, as we can see:

\begin{itemize}
\item[(iv)] $d(1,1,-1,0)=d(1,0,-1,1)=d(1,0,-1,-1)=d(1,-1,-1,0)=d(0,1,1,-1)=\linebreak d(0,-1,1,1)=d(1,1,0,-1)=d(1,-1,0,1)=d(0,1,-1,-1)=d(0,-1,-1,1)=\linebreak 
d(-1,1,0,-1)=d(-1,-1,0,1)=d(-1,1,1,0)=d(-1,-1,1,0)=d(-1,0,1,1)=\linebreak d(-1,0,1,-1)=3a^2-10b=A+C>A$;
\item[(v)] $d(\pm1,\pm1,0,0)=d(0,0,\pm1,\pm1)=d(1,-1,0,0)=d(-1,1,0,0)=d(0,0,1,-1)=\linebreak d(0,0,-1,1)=d(1,0,0,1)=d(-1,0,0,-1)=d(1,0,0,-1)=d(-1,0,0,1)=\linebreak d(0,1,1,0)=d(0,-1,-1,0)=d(0,1,-1,0)=d(0,1,-,1,0)=2(a^2-2b)=2A>A$;
\item[(vi)] $d(\pm1,\pm1,\pm1,0)=d(\pm1,\pm1,0,\pm1)=d(\pm1,0,\pm1,\pm1)=d(0,\pm1,\pm1,\pm1)= d(1,-1,1,0)=\linebreak d(-1,1,-1,0)=d(1,-1,0,-1)=d(-1,1,0,1)=d(1,0,1,-1)= d(-1,0,-1,1)=\linebreak d(0,1,-1,1)=d(0,-1,1,-1)=3a^2-2b=A+B>A$;
\item[(vii)] $d(\pm1,\pm1,\pm1,\pm1)=d(1,-1,1,-1)=d(-1,1,-1,1)=4a^2=2B>B$;
\item[(viii)] $d(1,1,-1,-1)=d(1,-1,-1,1)=d(-1,1,1,-1)=d(-1,-1,1,1)=2C>C$;
\item[(ix)] $d(1,1,1,-1)=d(-1,1,-1,-1)=d(1,-1,1,1)=d(-1,-1,-1,1)=d(1,1,-1,1)=\linebreak d(1,-1,-1,-1)=d(-1,1,1,1)=d(-1,-1,1,-1)=4A>A$.
\end{itemize}
According to analyses above, $|\Lambda_f|=A, B$ or $C,$ i.e.,
  $|\Lambda_f|=a^2-2b, 2a^2$ or $2a^2-8b.$
Note that if $a^2-2b=min\{a^2-2b,2a^2,2a^2-8b\}$, then $\Lambda_f$ is well-rounded, since the vectors  $x\in \Lambda_f$ such that $\Vert x \Vert ^2=a^2-2b$ are linearly independent. 
Let us show that $a^2-2b=min\{a^2-2b,2a^2,2a^2-8b\}$ is also a necessary condition for $\Lambda_f$ to be well-rounded. The proof is similar to that of the Theorem \ref{teograu3}.
Suppose that $\Lambda_f$ is well-rounded and  $a^2-2b\neq min\{a^2-2b,2a^2,2a^2-8b\}$. Clearly $b\neq 0.$
When $b>0$, it follows that $2a^2-8b<2a^2.$  By our assumption $2a^2-8b<a^2-2b$
and then we conclude that $2a^2-8b=min\{a^2-2b,2a^2,2a^2-8b\}$. Since  the vectors  $x\in \Lambda_f$ such that $\Vert x \Vert ^2=2a^2-8b$ are linearly dependent, it follows that $S(\Lambda_f)$ does not span $\mathbb{R}^4.$ On the other hand, 
when $b<0$, it follows that  $2a^2<2a^2-8b$. By our assumption, $2a^2<a^2-2b$ and then we conclude that
 $min \{a^2-2b,2a^2,2a^2-8b\}=2a^2$. Again, it is easy to see that in this case $S(\Lambda_f)$ does not span $\mathbb{R}^4.$
 Thus, $\Lambda_f$ is well-rounded if and only if $a^2-2b= min\{a^2-2b,2a^2,2a^2-8b\}$.  It means that  
$a^2-2b\leq 2a^2-8b$ and  $a^2-2b\leq 2a^2$.
When $b\geq 0$ ($b<0$)  the inequalities above are satisfied if and only if $a^2\geq 6b$ ($a^2\geq -2b$). Consequently, we have proved that the following theorem holds true.

\begin{theorem}\thlabel{teoremaprincipal3}
Let  $f(x)=x^4+ax^3+bx^2+cx+d\in \mathbb{Z}[x],$ with $a\neq 0$, be a polynomial  with distinct real roots $\alpha,\beta,\gamma,\psi$, where  $\alpha=-\gamma$. Under these conditions,
\begin{itemize}
\item[(i)] if $b<0$, then $\Lambda_f$ is well-rounded if and only if $a^2\geq-2b$;
\item[(ii)] if $b\geq  0$, then $\Lambda_f$ is well-rounded if and only if $a^2\geq 6b$.
\end{itemize} \end{theorem}

\begin{remark}
Although $|S(\Lambda_f)|$ increases when $a^2=-2b$ or $a^2=6b$, these situations do not give us the scenario where  $\Lambda_f$ has the highest packing density in dimension $4$, since in both cases $|S(\Lambda_f)|$ has $12$  minimal vectors,  while the ideal situation happens when the number of minimum vectors is $24$. This can also be verified by calculating the center density in these cases.
Thus, another approach is necessary to obtain 
lattices with the highest packing density, which will be treated in more details in a forthcoming paper.
\end{remark}

\section{Conclusion}

In this paper, we have investigated when lattices obtained by polynomials up to degree $4$ with integers coefficients are well-rounded. The construction of lattices via polynomials is not often exploited in the literature. The authors in \cite{andre} presented constructions of lattices up to dimension 3, however,  the property of well-roundedness was not treated.
The difficulty of this construction lies in finding linearly independent vectors consisting of the roots of a polynomial with integers coefficients. It is also a challenge to identify the  determinant of the generator matrix and the norm of a point in the lattice in terms of the coefficients of the polynomial considered. In dimension $4$, the hypothesis $\alpha=-\gamma$ was  fundamental to deal with these issues in our approach. Well-rounded lattices have been recently studied in several scenarios \cite{siso, mimo}, which makes this subject attractive and current.

\end{document}